\pdfoutput=1
\documentclass[10pt,conference,letterpaper]{IEEEtran}
\IEEEoverridecommandlockouts
\usepackage{cite}
\usepackage{balance}
\usepackage{amsmath,amssymb,amsfonts}
\usepackage{amsthm}
\usepackage{graphicx,subfig}
\usepackage[normalem]{ulem}
\usepackage{algorithm}
\usepackage{multirow}
\usepackage{amsopn}
\usepackage{etoolbox}
\usepackage{xcolor}
\usepackage{url}
\usepackage{xspace}
\usepackage{indentfirst}
\usepackage{microtype}
\usepackage{color}
\usepackage{listings}
\usepackage{cuted}
\usepackage{amsmath,lipsum}
\usepackage{tabularx}
\newcommand{\ourname}{{\sc SkyCastle}\xspace}
\usepackage[noend]{algpseudocode}
\usepackage{enumitem}
\usepackage{booktabs}
\usepackage{bm}

\newcommand{\name}{{\sc SkyCastle}\xspace}

\newcommand{\ie}{\emph{i.e.,}\xspace}
\newcommand{\eg}{\emph{e.g.,}\xspace}

\newtheorem{lemma}{\textbf{Lemma}}

\newtheorem{theorem}{\textbf{Theorem}}

\newcommand{\revise}[1]{\textcolor{black}{#1}}

\newcommand{\tabincell}[2]{\begin{tabular}{@{}#1@{}}#2\end{tabular}}

\begin{document}

\title{\name: Taming LEO Mobility to Facilitate Seamless and Low-latency Satellite Internet Services}

	\author{
	\IEEEauthorblockN{Jihao Li$^\dag$, Hewu Li$^\ddag$$^\ast$\textsuperscript{${^\P}$}, Zeqi Lai$^\ddag$$^\ast$\textsuperscript{${^\P}$}\thanks{${^\P}$~Hewu Li and Zeqi Lai are the corresponding authors.}, Qian Wu$^\ddag$$^\ast$, \\Weisen Liu$^\ddag$, Xiaomo Wang$^\dag$$^\pounds$, Yuanjie Li$^\ddag$$^\ast$, Jun Liu$^\ddag$$^\ast$, Qi Zhang$^\ast$}
	\IEEEauthorblockA{
		\textit{$^\dag$Department of Computer Science and Technology}, Tsinghua University\\ \textit{$^\ddag$Institute for Network Sciences and Cyberspace}, Tsinghua University\\
		\textit{$^\ast$Zhongguancun Laboratory}
		\quad \quad \textit{$^\pounds$China Academy of Electronics and Information Technology}
	}
}

	\maketitle
	
	\begin{abstract} 

Emerging integrated space and terrestrial networks (ISTN) built upon low earth orbit (LEO) satellite constellations aim at providing planet-wide Internet services, not only for \emph{residential users}, but also for \emph{mobile users}~(\eg in airplane and cruise scenarios). Efficiently managing \emph{global mobility} and keeping connections active for mobile users is critical for ISTN operators. However, our quantitative analysis identifies that existing mobility management~(MM) schemes suffer from frequent connection interruptions and long latency in ISTN scenarios. The fundamental challenge stems from a unique characteristic of ISTNs: not only users are mobile, but also \emph{core network infrastructures}~(\ie LEO satellites) are frequently changing their locations in the network.

To facilitate \emph{seamless and low-latency} satellite Internet services, this paper presents \name, a novel network-based \emph{global mobility management mechanism}. \name incorporates two key techniques to address frequent connection interruptions in ISTNs. First, to reduce the interruption time, \name adopts \emph{distributed satellite anchors} to track the location changes of mobile nodes, manage handovers and avoid routing convergence. Second, \name leverages an \emph{anchor manager} to schedule MM functionalities at satellites to reduce deployment costs while guaranteeing low latency. Extensive evaluations combining real constellation information and mobile user trajectories show that: \name can improve up to 55.8\% uninterrupted time and reduce 47.8\% latency as compared to other existing MM solutions.

\end{abstract}
	
	\begin{IEEEkeywords}
		Integrated Space and Terrestrial Networks, Satellite Internet Services, Global Mobility Management.
	\end{IEEEkeywords}
	
	\section{Introduction}
\label{sec:introduction}

\columnsep 0.12in

With the rapid development of aerospace technology, recently we have witnessed the vigorous deployment of satellite Internet mega-constellations such as SpaceX's Starlink~\cite{starlink}, Amazon's Kuiper~\cite{kuiper} and Telesat~\cite{telesat}. These constellations will deploy thousands of broadband satellites in low earth orbit~(LEO), and each satellite can be equipped with high-speed inter-satellite links~(ISLs)~\cite{laser2} as well as ground-satellite links~(GSLs)~\cite{KAKU} to interconnect with other satellites and ground facilities, extending the boundary of today’s Internet and constructing an integrated space and terrestrial network~(ISTN).

Upon the IP-based infrastructure~\cite{starlink_IP}, future ISTNs will support heterogeneous access technologies~(\eg Ka-/Ku-/E-band~\cite{KAKU} or personal communications services for phone-direct-to-satellite access~\cite{TMobile}) for global Internet services. ISTN users include not only residential customers who leverage fixed satellite terminals deployed on rooftops~\cite{starlink} to access satellites, but also global mobile users in airplanes~\cite{starlink_airline} and maritime cruises~\cite{starlink_maritime} which can roam globally over time. In practice, Starlink has recently released its global roaming services to enable Internet services all around the world~\cite{starlink_roam}. How to efficiently manage the mobility of mobile users~(\eg tracking the location of users and coping with handovers) and keep users' connections active is an important issue for ISTN operators.

Network-layer mobility management~(MM) has always been an important technique for managing user mobility, especially for heterogeneous mobile networks because it does not rely on or make any assumption about the underlying wireless access technologies. The network community has a long history studying on network-layer MM solutions such as IETF MIP~\cite{MIP,MIPv6} and its variations~\cite{HMIPv6,PMIPv6,FMIPv6}. The core idea behind existing solutions is to deploy \emph{anchor points} at certain fixed nodes in the network. Each mobile node first registers with a certain anchor point and then informs the anchor of its location updates during the roaming process. In addition, each mobile node is identified by its \emph{home address} allocated by its corresponding anchor, which further allows users to move from one network to another while maintaining a permanent address to keep connections active.

However, as our quantitative analysis shows in \S\ref{sec:motivation}, directly applying existing MM solutions in emerging ISTNs may suffer from \emph{low connection uninterrupted ratio} or \emph{high user-perceived latency}. The root cause relies on a unique characteristic differentiating ISTNs from other traditional mobile networks: not only the end users are mobile, but also the global network infrastructures~(\ie LEO satellite routers) which participate in route calculation are frequently changing their locations in the network. In an ISTN, if we directly follow existing solutions and deploy anchors on the ground~(\eg at ground stations), the combination of space-ground handovers and routing fluctuations can exacerbate service interruptions when users are roaming around the world. Moreover, as both users and satellites are moving, the user-to-anchor path may be lengthened, causing a significant increase in user-perceived latency over time.

To facilitate seamless and low-latency satellite Internet service, in this paper we present \name, a novel network-layer global mobility management mechanism for futuristic IP-based ISTNs. At a high level, \name divides all satellites into multiple \emph{clusters}, and exploits a collection of \emph{distributed satellite anchors} to manage the mobility of the users and the core network infrastructures. In particular, \name incorporates two techniques for global mobility management.

First, \name adopts a dynamic-anchor-based MM scheme, together with a convergence-free route mechanism to efficiently track the location of mobile nodes, while providing available forwarding paths when the ISTN topology changes. For users covered by the same cluster and managed by the same satellite anchor, they register with the anchor when they connect to the current cluster, and inform the anchor of their location updates if their locations change. Moreover, the location updates of a ground station are additionally sent to all other anchors in the ISTN to avoid route convergence.

Second, \name employs an \emph{anchor manager} integrating a series of algorithms to judiciously decide how a satellite operator can distribute anchor functionalities from a constellation perspective to reduce the number of interruptions, bound latencies to satisfy various user requirements, and reduce the deployment cost of satellite anchors.

To evaluate the effectiveness of \name, we build an experimental ISTN environment based on a recent simulator~\cite{starrynet} for satellite networks, and implement a prototype containing all core functionalities of \name. Through trace-driven evaluations combining real constellation information, mobile user trajectories~(\eg popular flight routes) and network simulation, we demonstrate that for representative global roaming scenarios in ISTNs, \name can: (i) improve the connection uninterrupted time by up to 55.8\% and by 34.5\% on average; (ii) decrease the user-perceived latency by up to 47.8\% and by 21.5\% on average, as compared to existing solutions.

\begin{figure}[tbp]
	\centerline{\includegraphics[width=0.9\linewidth]{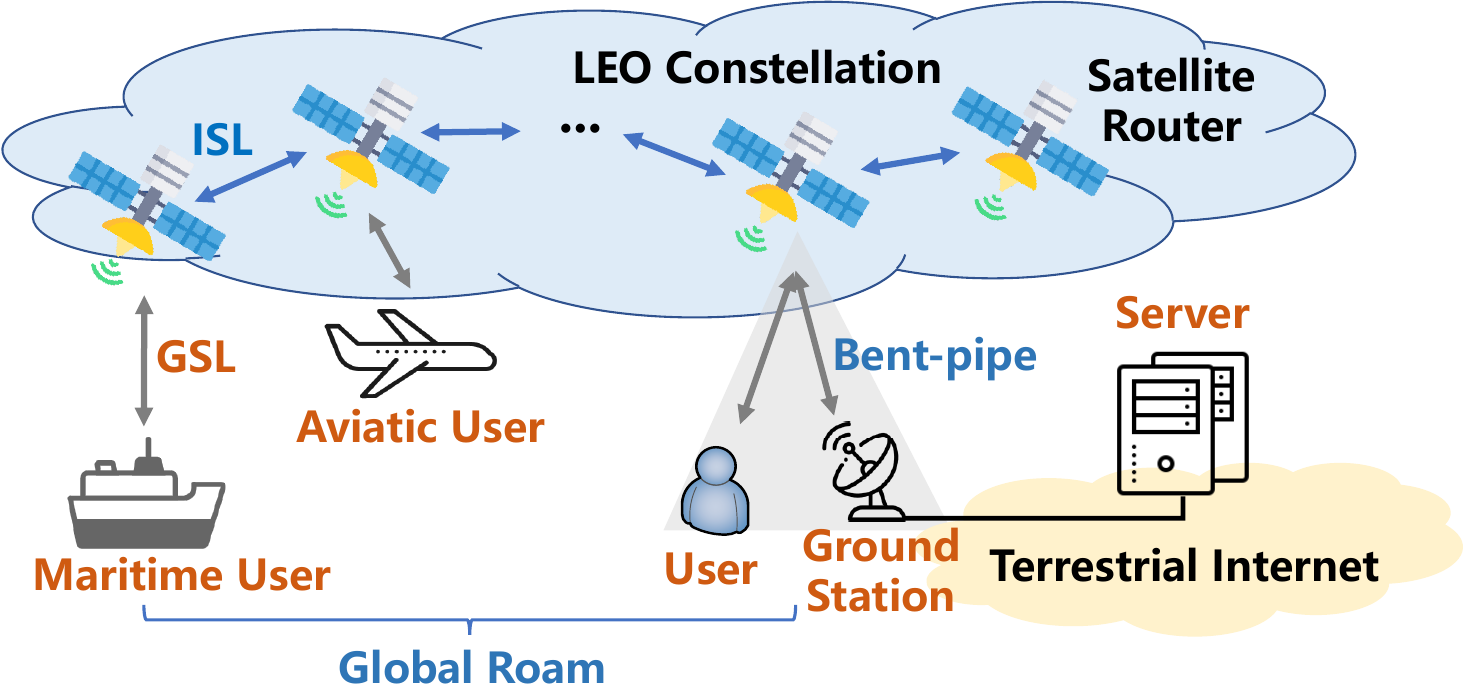}}
	\caption{A typical architecture of ISTN.}
	\label{fig:bent_pipe}
	\vspace{-0.6cm}
\end{figure}

Contributions of this paper can be concluded as follows.

\begin{itemize}[leftmargin=*]
	
	\item We expose and analyze the challenges for global mobility management in emerging ISTNs, where both end users and network infrastructures are inherently mobile~(\S\ref{sec:motivation}).
	
	\item We propose \name, a novel global mobility management mechanism exploiting dynamic and distributed satellite anchors to facilitate seamless and low-latency satellite Internet services globally~(\S\ref{sec:design_overview},\S\ref{sec:design_at_the_anchor_level},\S\ref{sec:design_at_the_network_level}). 
	
	\item We build a \name prototype and conduct trace-driven simulations to show its effectiveness on improving network availability for global roaming in ISTNs~(\S\ref{sec:evaluation}).
\end{itemize}

\addtolength{\topmargin}{0.05in}

\section{Technical Background}
\label{sec:technical_background}

\subsection{Integrated Space and Terrestrial Networks~(ISTNs)}
\label{subsec:leo_satellite_network}
 
Figure~\ref{fig:bent_pipe} plots a high-level architecture of emerging ISTNs upon LEO satellite constellations, and this architecture has been used by today's operational ISTNs such as Starlink~\cite{starlink}. In particular, LEO satellites are moving rapidly related to the earth surface. Each satellite can be equipped with laser inter-satellite links~(ISLs) to interconnect with each other, and radio ground-satellite links~(GSLs) to communicate with ground facilities such as satellite terminals and ground stations~(GSs). When users access Internet services via the ISTN, packets from the user side are first sent to a ground station through one or multiple satellites~(\ie by ``bent-pipe'' transparent forwarding if the user is close to the ground station, or by multi-hop routing for remote users), and then forwarded to the Internet by terrestrial fibers~\cite{groundpop}. Emerging ISTNs are aimed at providing Internet services to diverse users, such as residential customers~\cite{starlink}, and mobile users in recreational vehicles~(RVs)~\cite{starlink_mobility}, airplane~\cite{starlink_fly} and maritime scenarios~\cite{starlink_maritime}. In practice, recently Starlink releases its global roaming service~\cite{starlink_roam} to enable ubiquitous access on the earth via ISLs, and provides low-latency in-flight Internet across the globe~\cite{starlink_available_flight}.

\subsection{Network-layer Mobility Management~(MM)}
\label{subsec:ip_mobility_management}

\begin{figure}[tbp]
	\centerline{\includegraphics[width=0.79\linewidth]{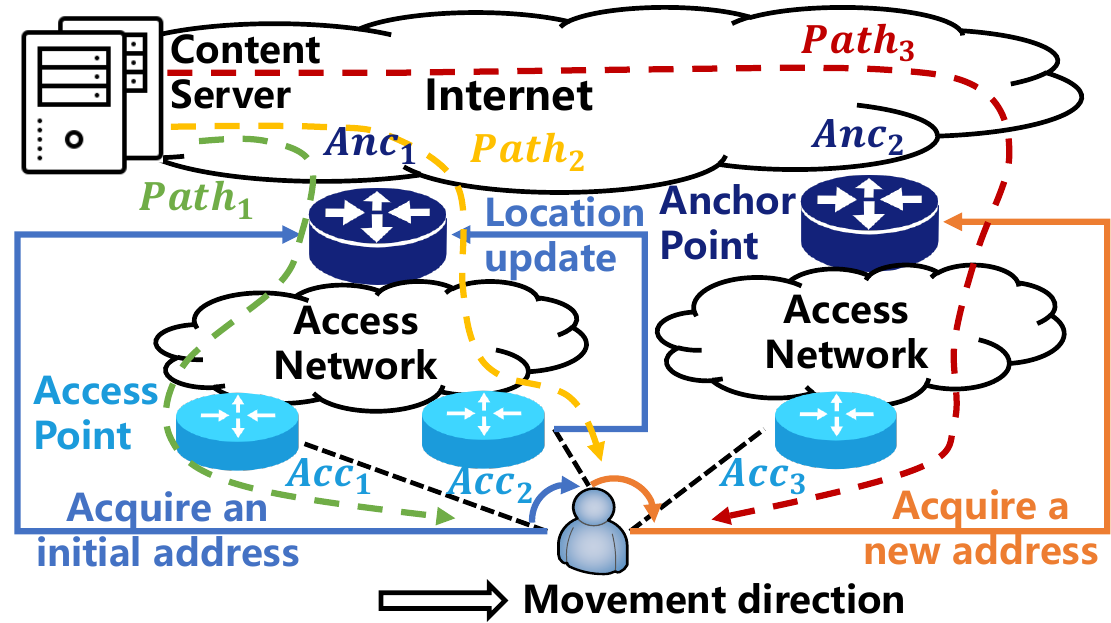}}
	\caption{Basic principles for mobility management~(MM).}
	\label{fig:MM}
	\vspace{-0.55cm}
\end{figure}

In global roaming scenarios, since users are continuously changing their locations in the network, it is important to track where users are, and deal with handovers to keep connections active. Network-layer mobility management~(MM), in which mobility-related features are deployed at the IP layer and signaling messages for mobility purposes are carried by IP traffic, is a classic approach to handle mobility issues, and has been well studied over the past decade in terrestrial Internet~\cite{MIP,HMIPv6,FMIPv6,PMIPv6,PFMIPv6,DMM,RDMM}. As plotted in Figure~\ref{fig:MM}, the core idea of network-layer MM is to exploit \emph{anchor points} to manage the locations and handovers for mobile users. Specifically, an anchor is a fixed node~(\eg a router or switch) maintaining the mobility status of users. When a mobile user connects to the network through an access point ($Acc_1$), it first registers with a specific anchor ($Anc_1$) and acquire an initial address. Traffic from or to the mobile user is first forwarded to the anchor and then to the final destination ($Path_1$). When the user changes location in the network (\ie connects to $Acc_2$) within the management area of the registered anchor ($Acc_1$), $Acc_2$ sends messages to $Anc_1$ and notifies the location change in real time. Correspondingly, the traffic routes according to $Path_2$ based on the new location information in $Anc_1$. In particular, once the user is managed by a new anchor ($Anc_2$), it needs to re-register and acquire a new address. Similarly, the traffic will pass through the new anchor point ($Path_3$).

\noindent
\textbf{Research scope.} In this paper, we focus on exploring the network-layer mobility management for ISTNs for three key reasons. First, to efficiently integrate satellite constellations into existing terrestrial Internet, emerging ISTNs like Starlink adopts IP-based networking architecture~\cite{starlink_IP} with heterogeneous wireless access technologies. Second, unlike link-layer solutions~(\eg~\cite{linkMM}), a network-layer solution does not rely on any assumption about the underlying wireless access technologies. Finally, as compared to high-layer approaches~(\eg~QUIC's connection identifier~\cite{quic}), a network-layer solution achieves higher handover efficiency~(\eg shorter disruption time). Collectively, solutions in other layers complement our work.

\begin{figure}[t]
	\centering
	\subfloat[Deploying anchors at GSs.]{
		\begin{minipage}[t]{0.448\linewidth}
			\centering
			\includegraphics[width=1.0\linewidth]{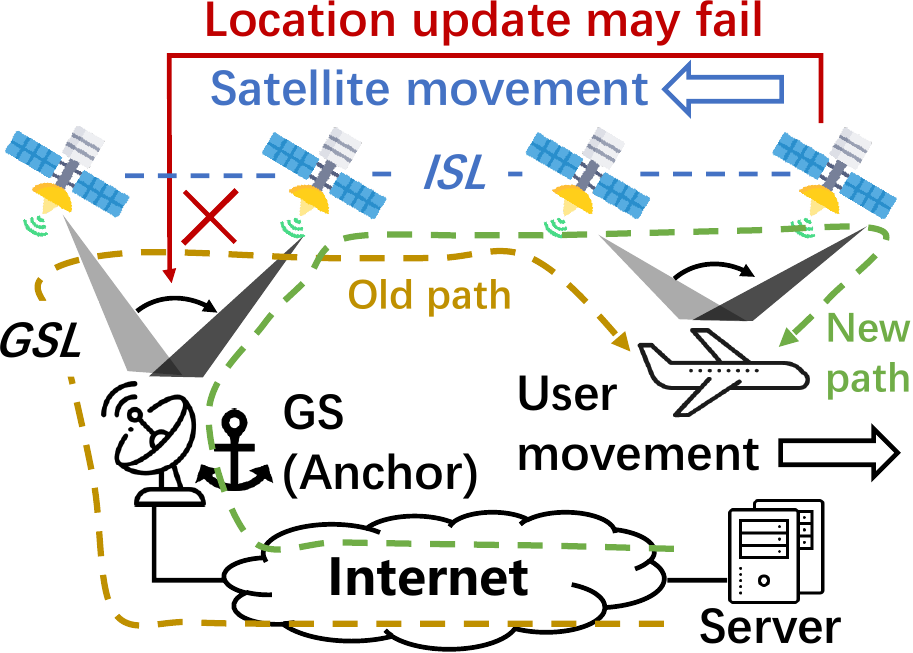}
			\label{fig:CMM_LSN}
			\vspace{-0.4cm}
		\end{minipage}%
	}\hspace{0.01cm}
	\subfloat[Deploying anchors at satellites.]{
		\begin{minipage}[t]{0.508\linewidth}
			\centering
			\includegraphics[width=1.0\linewidth]{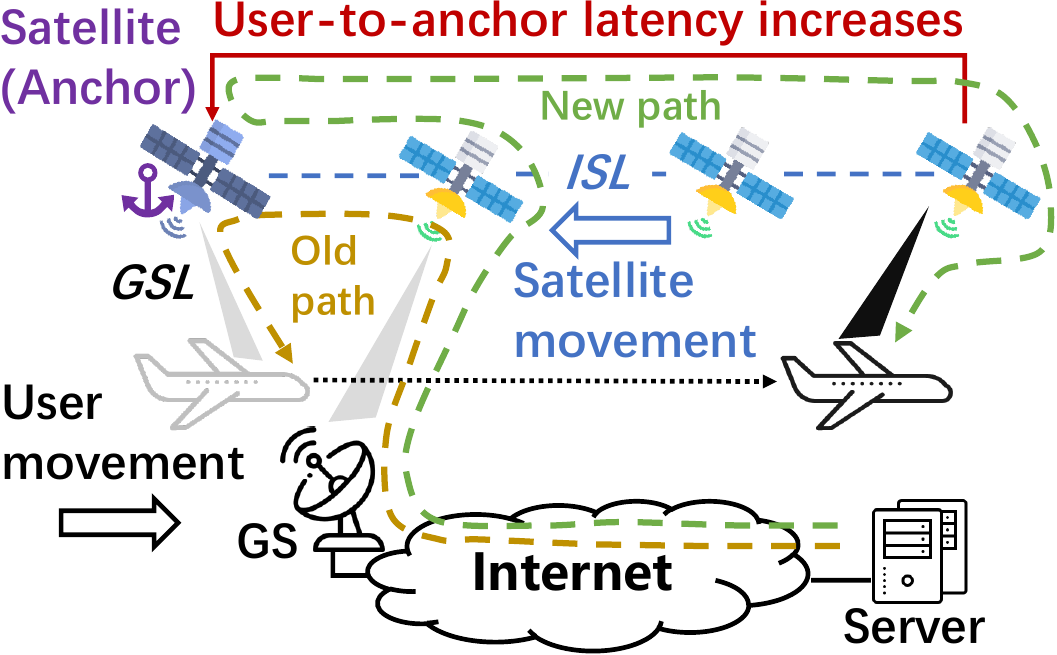}
			\label{fig:DMM_LSN}
			\vspace{-0.4cm}
		\end{minipage}%
	}%
	\centering
	\caption{The design choices of deploying MM in ISTNs.}
	\label{fig:MM_ISTN}
	\vspace{-0.6cm}
\end{figure}

\section{Quantitative Performance Analysis of Deploying Mobile Anchors in ISTNs}
\label{sec:motivation}

Given the unique characteristics of futuristic ISTNs and the basic principles of network-layer MM approaches above, we thus ask the question: how should satellite operators deploy MM solutions in their ISTNs to sustain good user experience under both LEO dynamics and user movements?

\subsection{Analysis Experiment Setup}
\label{subsec:motivation_experiment_setup}

We start our quest by quantitatively analyzing the attainable performance of directly applying existing MM solutions in ISTNs. We build an ISTN simulator based on~\cite{starrynet}, an open-source simulation tool for satellite networks. The ISTN simulator can: (i) mimic the LEO dynamics as well as the corresponding network behaviors~(\eg link connection/disconnection); and (ii) run routing protocols in simulated routers~(\ie satellites and GSs) and (iii) load different MM approaches. Specifically, we simulate an ISTN based on the public information of the first shell of Starlink~\cite{starlink_fcc}, with 1584 LEO satellites in the well-known +Grid topology~\cite{motif,hypatia,DPR} and about 200 distributed GSs~\cite{starlink_GS}. Each satellite connects to the two front and rear satellites in the same orbit and two left and right satellites in its adjacent orbits to construct a gird-like structure.

\noindent
\textbf{Metrics.} Ideally, an MM approach should accomplish: (i) \emph{seamless handover}, which indicates that the connection interruption time should be short during handovers; and (ii) \emph{low latency}, \ie the user perceived latency after handovers is expected to be limited in an acceptable range. Accordingly, we consider two key performance metrics: (i) \emph{end-to-end connection uninterrupted ratio~(CUR)}, which is defined as the percentage of time that the network connection between the server and user is available, and (ii) \emph{round-trip time~(RTT)}.

\subsection{Deploying Anchors at GSs Suffers from Frequent and Long Connection Interruptions}
\label{subsec:motivation_anchor_on_GSs}

A straightforward approach of applying existing MM approaches in ISTNs is to deploy anchors at some fixed nodes in the terrestrial network~(\eg at GSs), as shown in Figure~\ref{fig:CMM_LSN}. In this case, a mobile user~(\eg airplane) registers with an anchor at a nearby GS. When the user changes their location, the satellite to which the user connects~(\ie the ingress satellite) sends notifications to notify the anchor where the user is. Traffic from the server is first forwarded to the GS anchor of the user and then to the user's ingress satellite, and finally to the user.

\noindent
\textbf{Observations.} Figure~\ref{fig:reliability_motivation} and~\ref{fig:delay_motivation} plot the end-to-end CUR and RTTs when anchors are deployed at GSs. Our experiment results reveal that deploying anchors on GSs can result in poor network continuity. As shown in Figure~\ref{fig:reliability_motivation}, The server-to-user CUR is only 79\% on average, indicating that the connection is unavailable for at least 20\% of the entire session.

\noindent
\textbf{Root cause.} On our in-depth analysis, we find that the low connection uninterrupted ratio is caused by the \emph{routing instability} between users' ingress satellites and corresponding anchors at GSs due to space-ground handovers. As LEO satellites move, the ISTN topology changes frequently, which results in routing fluctuations and temporal route unreachability. As shown in Figure~\ref{fig:CMM_LSN}, after the handover, \emph{the route from the ingress satellite to GS needs to be recalculated}. During this routing convergence period, the location update may fail and the GS does not know the latest location of the user. Consequently, the server cannot correctly send traffic to the mobile user.

\begin{figure}[tbp]
	\centering
	\subfloat[CUR.]{
		\begin{minipage}[t]{0.18\linewidth}
			\centering
			\includegraphics[width=1.0\linewidth]{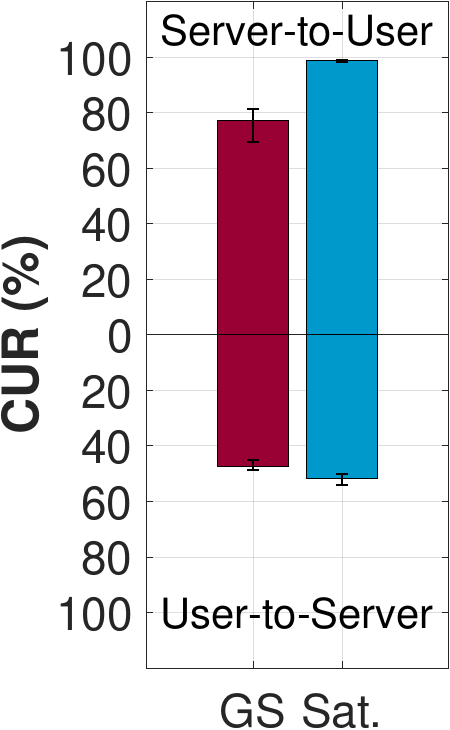}
			\label{fig:reliability_motivation}
			\vspace{-0.4cm}
		\end{minipage}%
	}
	\vspace{-0.2cm}
	\subfloat[RTT.]{
		\begin{minipage}[t]{0.23\linewidth}
			\centering
			\includegraphics[width=1.0\linewidth]{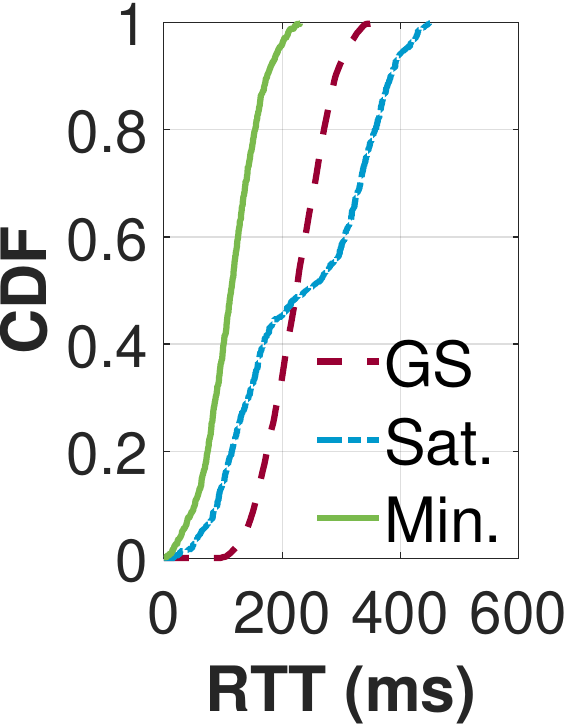}
			\label{fig:delay_motivation}
			\vspace{-0.4cm}
		\end{minipage}%
	}
	\subfloat[RTT increases over time.]{
		\begin{minipage}[t]{0.53\linewidth}
			\centering
			\includegraphics[width=1.0\linewidth]{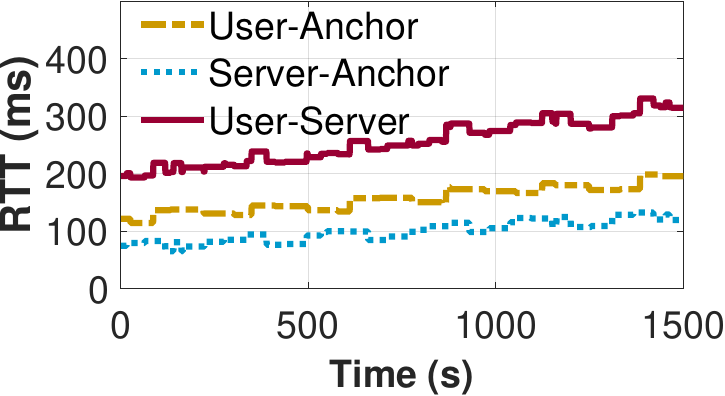}
			\label{fig:vary}
			\vspace{-0.4cm}
		\end{minipage}%
	}%
	\centering
	\caption{Performance of applying anchors at GSs and satellites.}
	\label{fig:motivation}
	\vspace{-0.5cm}
\end{figure}

\subsection{Deploying Anchors at Satellites is Latency-limited}
\label{subsec:motivation_anchor_on_sat}

Another viable path to tackle mobility is to deploy anchors at access points for distributed mobility management~\cite{DMM}, \ie deploying anchors at the ingress satellites. As shown in Figure~\ref{fig:DMM_LSN}, in this case, each satellite works not only as a router to forward network traffic, but also as an anchor to manage the location and handover of users that register with it.

\noindent
\textbf{Observations.} Figure~\ref{fig:reliability_motivation} and~\ref{fig:delay_motivation} plots the end-to-end CUR and RTTs when users choose their \emph{initial} ingress satellite as the anchor. We observe that deploying anchors on satellites can significantly improve the server-to-user CUR, but it still suffers from high user-perceived RTTs. Figure~\ref{fig:vary} plots the RTTs changing over time. As users roam around the world and satellites move around the earth, we observe that the RTT can increase by about 100ms over a period of 25 minutes.

\noindent
\textbf{Root cause.} While LEO satellites move at a high velocity relative to the earth's surface,  the relative positions of the front/rear satellites in the same orbit, and those of the left/right satellites in adjacent orbits, can remain stable~\cite{motif}, because they operate at the same orbital altitude and move at the same velocity. Therefore, although the routes from satellites to GSs suffer from instability, \emph{the routes between satellites are relatively stable} due to fixed inter-satellite topology. As a result, when anchor points are deployed at satellites, the ingress satellite can always correctly and efficiently update users' location to the anchor. In this case, packets are sent from the GS to satellite anchor firstly, and then forwarded to the user's latest ingress satellite according to anchor's location information. Note that although the optimal path from the GS to the satellite anchor also needs to be recalculated, the GS can send packets to any satellite it is currently connected to. This satellite always maintains a stable and valid route to the satellite anchor, ensuring that there is an available route from the GS to the anchor. However, the latency can significantly increase as both the anchor and the user move at a high velocity, which prolongs the user-to-anchor path.

\subsection{Takeaways}
\label{subsec:motivation_takeaways}

In addition to the above analysis, our observations from Figure~\ref{fig:reliability_motivation} indicate that all current MM mechanisms achieve only about 50\% of the user-to-server CUR. The primary reason for this is that the routing instability, as mentioned above, not only leads to the failure of location updates but also causes interruptions in the data traffic from users to servers. Therefore, our quantitative analysis exposes the problems and challenges of conducting mobility management in ISTNs as follows.
\begin{itemize}[leftmargin=*]
	\item Since emerging ISTNs will provide Internet services for mobile customers such as airplanes and cruise users in motion, it should be  important to effectively manage mobility and keep connections active during handovers in ISTNs.
	
	\item The key of MM is to leverage an \emph{anchor point} to track user's time-varying location and forward network traffic. However, naively deploying anchors either suffers from long connection interruption time~(\eg deploying anchors at GSs), or can result in high latency~(\eg deploying anchors at satellites).
	
	\item Our in-depth analysis reveals that the root cause for the above inefficiency is that current MM approaches only consider the user-level mobility but ignore the unique infrastructure-level mobility in an ISTN environment, which results in routing instability and long latency between users and anchors.
\end{itemize}


	\section{\name Overview}
\label{sec:design_overview}

\subsection{Key Ideas Behind \name}

We present \name, a novel global mobility management mechanism with the key ideas as follows.

\begin{itemize}[leftmargin=*]
	
	\item \textbf{(i) Judiciously deploying distributed anchors at LEO satellites to efficiently manage the mobility of both users and network infrastructures.} 
	Taking the satellite constellation as the frame of reference, users and GSs continuously change their topological locations in the ISTN. \name leverages \emph{anchors at satellites} to track where a user or GS is and cope with their handovers. Specifically, when a user or GS connects to a new ingress satellite, its new ingress satellite sends an MM message to its satellite anchor to notify its current location. The ingress-to-anchor path is stable since the network topology of the space segment in the same orbital altitude is stable. 
	
	\item \textbf{(ii) Dynamically assigning satellite anchors to users to avoid triangular routing and attain low latency.} Because satellite anchors are moving, the user-to-anchor path can be prolonged over time. To avoid the latency increase, \name dynamically updates the anchor assignments and always chooses anchors satisfying the latency requirements while avoiding frequent anchor changes for users.

\end{itemize}

\noindent
\subsection{\name Architecture}

Figure~\ref{fig:overview} plots the high-level overview of our \name design. A mobile user~(\eg an airplane) visits its destination content server via a path built upon a sequence of LEO satellites and GSs. Based on this typical architecture, \name incorporates two core components: (i) a collection of \emph{satellite anchors} which are deployed at a subset of LEO satellites in the ISTN to handle the mobility; and (ii) an \emph{anchor manager} deployed at the satellite operator's control center which runs a series of algorithms to dynamically decide which satellites should carry the anchor functionality, and decide how geo-distributed mobile users should be assigned to different anchors.

\begin{figure}[tbp]
	\centerline{\includegraphics[width=\linewidth]{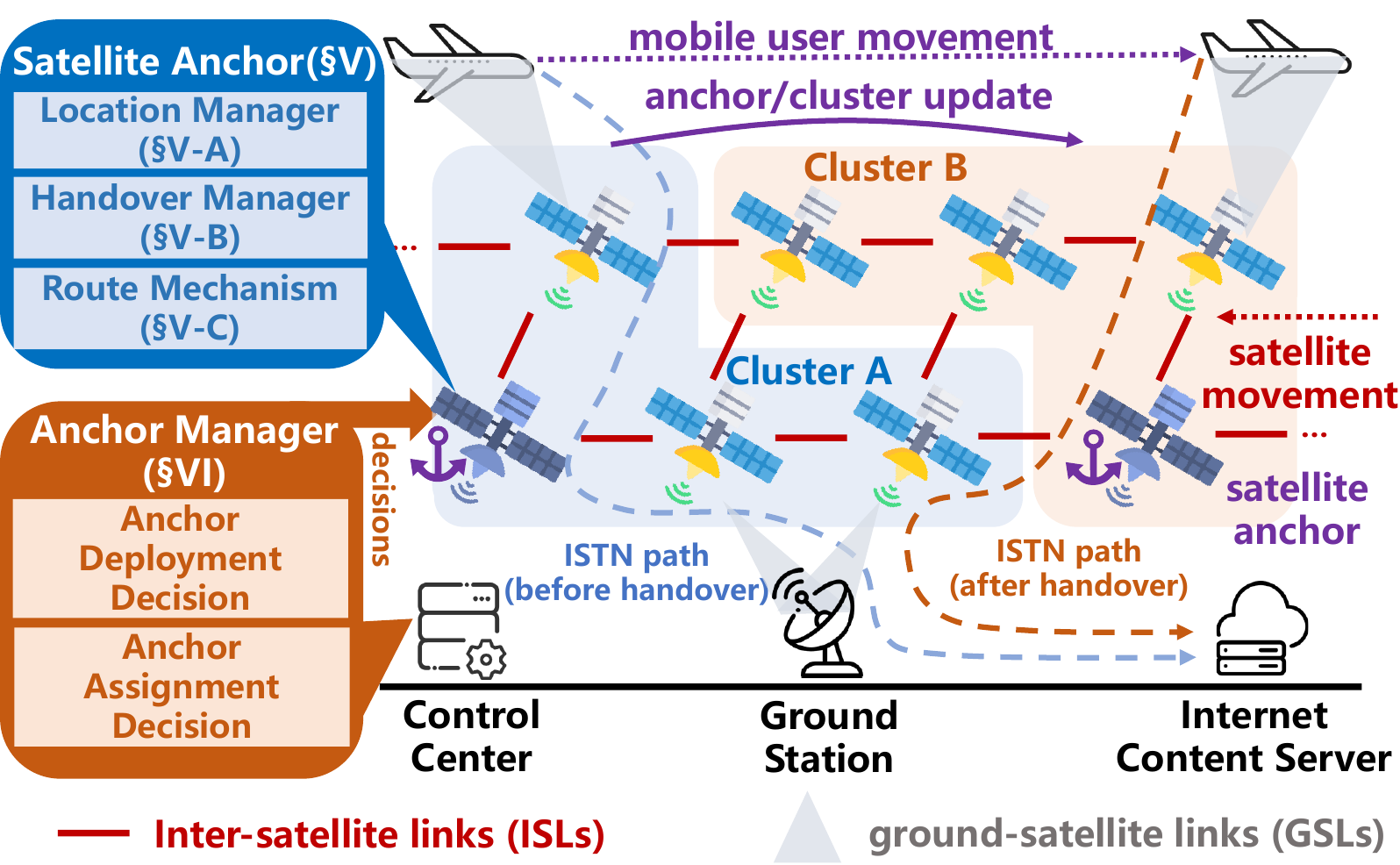}}
	\vspace{-0.05in}
	\caption{Overview of \name.}
	\label{fig:overview}
	\vspace{-0.2in}
\end{figure}

\noindent
\textbf{Satellite anchor and cluster division~(\S\ref{sec:design_at_the_anchor_level}).} \name divides all satellites into multiple disjoint \emph{clusters}. Each cluster is a satellite set which contains a collection of geographically adjacent satellites. One satellite in a cluster is selected as the \emph{anchor point of this cluster}. A satellite anchor manages the mobility of all GSs and mobile users that connect to satellites in this cluster. In addition to managing the location~(\S\ref{subsec:location_management}) and handover~(\S\ref{subsec:design_handover_management}) for users and GSs, an anchor also adopts a convergence-free route mechanism ~(\S\ref{subsec:design_route_management}) to deal with routing fluctuations and accomplish fast route recovery during LEO dynamics to reduce the connection interruption time. 

\noindent
\textbf{Anchor manager~(\S\ref{sec:design_at_the_network_level}).} Since both users and satellites are mobile, a user may change their current cluster as well as their corresponding satellite anchor. Thus, the goal of \name's anchor manager is to judiciously place anchor functionalities at different satellites and decide proper anchor assignments for geo-distributed users to guarantee that the end-to-end latency would not exceed a performance threshold specified by the satellite operator, while ensuring infrequent anchor/cluster change to reduce the interruption time for mobile users.

\noindent
\textbf{Working stages.} We describe an example to illustrate the working stages of \name. First, as shown in Figure~\ref{fig:overview}, the operator invokes the anchor manager to calculate and decide how to divide all satellites into multiple clusters and which satellite is selected to deploy anchor functionalities. These decisions are delivered to satellites via distributed ground stations. When a user connects to the ISTN, it registers the current location information with the anchor. The anchor then assigns a unique IP address to the user. Interactive traffic between the user and their target server is forwarded via the anchor. Due to the mobility of both satellites and users, a user may change to another cluster, which triggers an anchor update. After the cluster handover, interactive traffic between the user and server is carried over a new path via the new anchor.

	\section{\name at the Anchor Level: Core Functionalities for ISTN Mobility Management}
\label{sec:design_at_the_anchor_level}

\subsection{Location Management}
\label{subsec:location_management}

The goal of location management is to track the location of mobile nodes in a network. The key difference between a user and GS node is that: GS nodes fundamentally belong to network infrastructure nodes~(\eg routers) and run routing protocols to calculate routes, while user nodes are located at the edge of networks and do not participate in routing calculation. In addition, users' IP addresses are allocated by their anchor points while the IP addresses of GSs are configured by the satellite operator and they are fixed and do not change. In \name, the location information of a node is denoted as a key-value tuple, where the key is the IP address of the node and the value is the IP address of its ingress satellite. We next describe the location management for mobile user nodes and GS nodes respectively, which is also plotted in Figure~\ref{fig:handover}.

\begin{figure}[tbp]
	\centerline{\includegraphics[width=\linewidth]{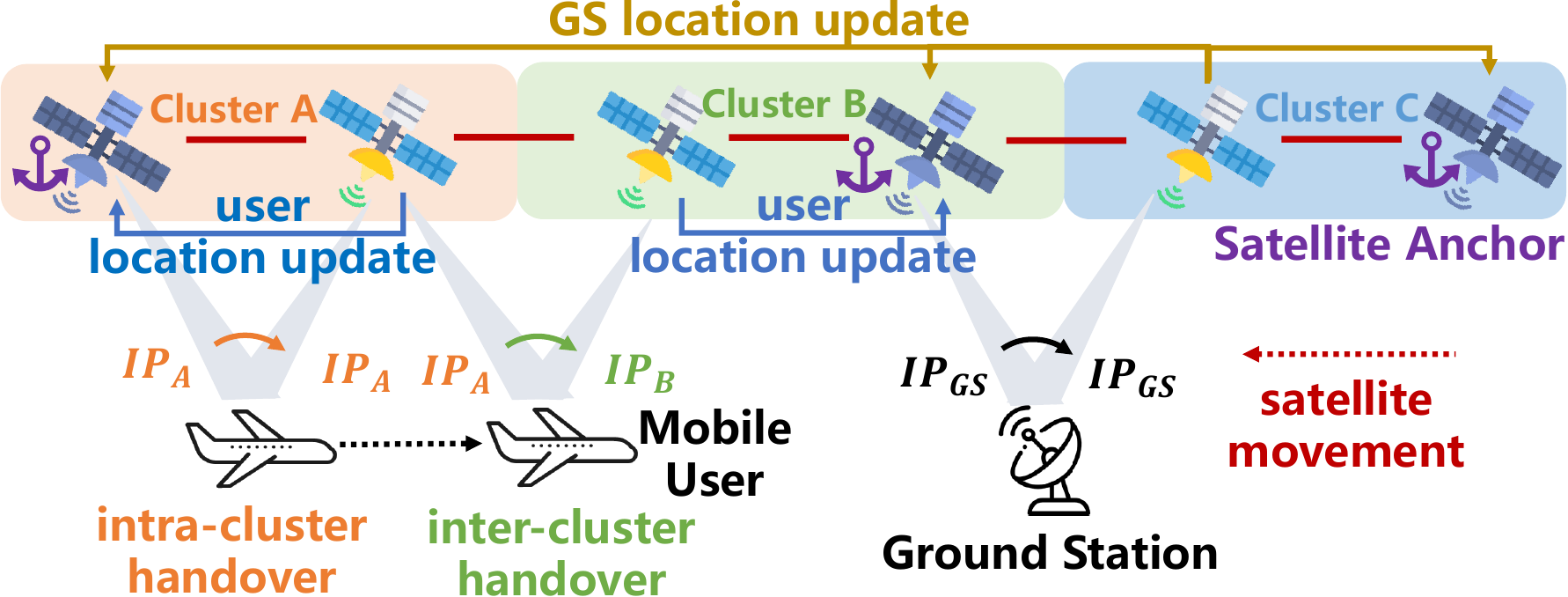}}
	\caption{Mobility management in \name. Each ingress satellite sends the location update message of mobile users to the anchor in current cluster, and sends the location update message of GSs to all anchors in the ISTN.}
	\label{fig:handover}
	\vspace{-0.5cm}
\end{figure}

\noindent
\textbf{Location management for users.} The IP address of a node implicitly indicates its location in the network. Once a mobile user node connects to the ISTN, its ingress satellite generates the location information of the user, and registers with the anchor node. If the user moves and changes its ingress satellite, the new ingress satellite updates the location information of the user and notifies the anchor point in the new cluster.

\noindent
\textbf{Location management for GSs.} When a GS node changes its position in the network, its ingress satellite notifies \emph{all anchor points} with the latest location information of the GS node. By this method, all anchors in the ISTN can track the position of all GSs as well as their ingress satellites in real time, and thus avoid the routing convergence process when the network topology changes. When a satellite needs to forward traffic to a terrestrial server, data packets are first forwarded to its anchor point, then to the GS, and finally to the destination.

\subsection{Handover Management}
\label{subsec:design_handover_management}

\noindent
\textbf{Intra-cluster handover.} For a mobile user, their IP address is allocated by the corresponding anchor when it connects to the cluster and registers with the anchor. As the mobile user moves within a cluster, their anchor point, as well as allocated address, remains unchanged, ensuring the user's higher-layer connections~(\eg TCP connections) stay active and stable.

\noindent
\textbf{Inter-cluster handover.} Because satellites are inherently mobile, the user-to-anchor path might lengthen over time and increase latency. To avoid the excessive latency caused by far anchor points, \name triggers a cluster/anchor update to guarantee the user-perceived latency is within an expected threshold. The new anchor in the new cluster allocates a new address to the user. \name's anchor manager takes care of the anchor assignment and guarantees user-perceived latency while avoiding inter-cluster handovers as much as possible.

\begin{figure}[tbp]
	\centering
	\subfloat[Route from a GS to a user.]{
		\begin{minipage}[t]{0.505\linewidth}
			\centering
			\includegraphics[width=1.0\linewidth]{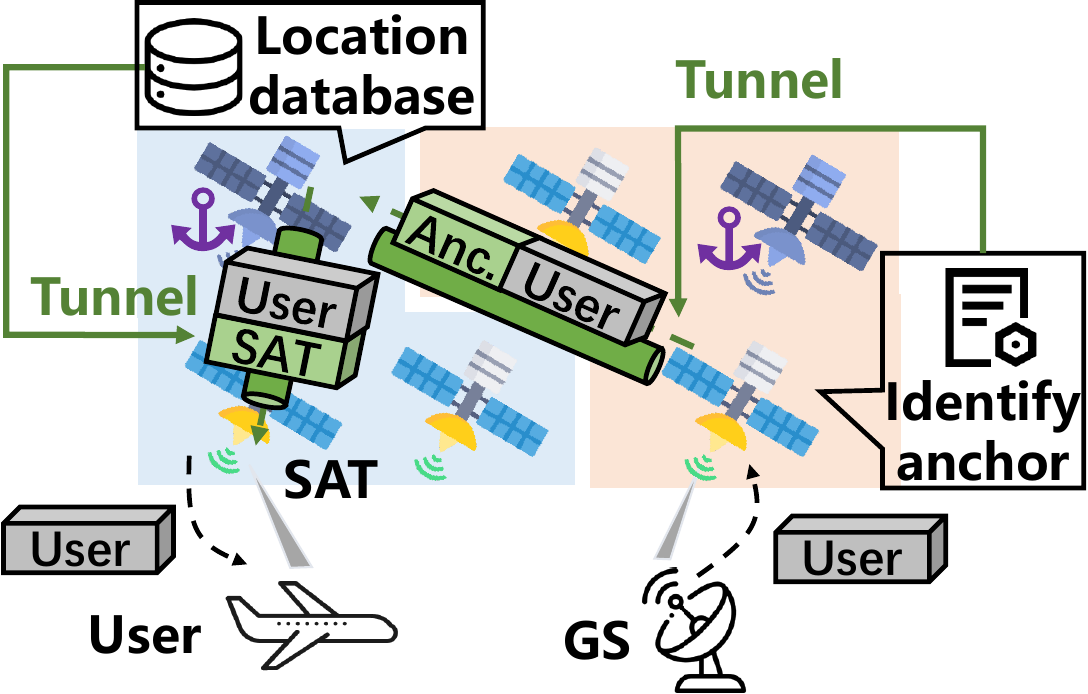}
			\label{fig:GS-to-MN}
			\vspace{-0.3cm}
		\end{minipage}%
	}\hspace{0.1cm}
	\subfloat[Route from a user to a GS.]{
		\begin{minipage}[t]{0.435\linewidth}
			\centering
			\includegraphics[width=1.0\linewidth]{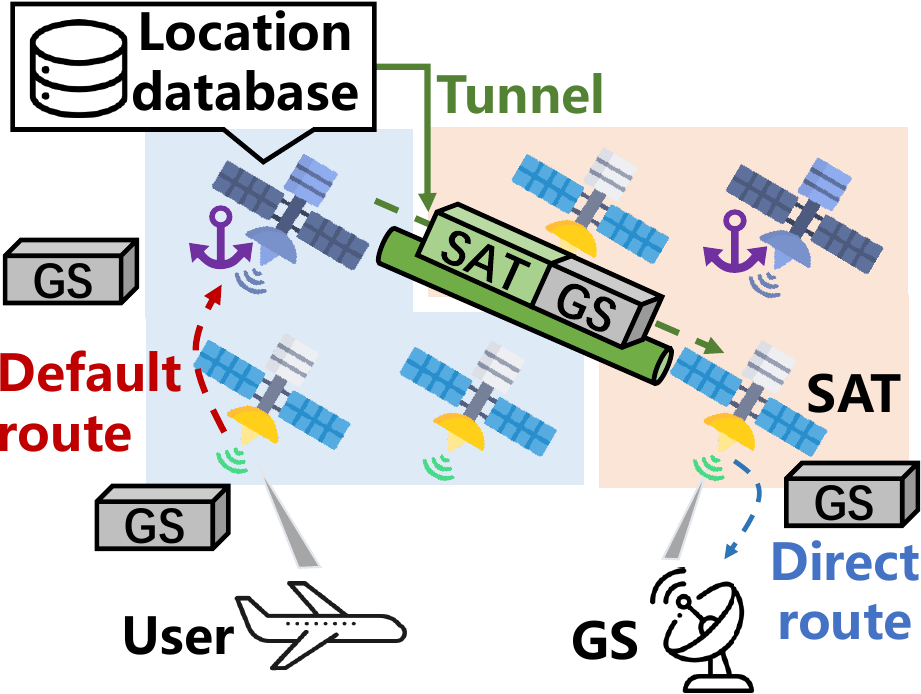}
			\label{fig:MN-to-GS}
			\vspace{-0.3cm}
		\end{minipage}%
	}%
	\centering
	\caption{Passing-anchor route in \name. The information in the cuboid represents the destination of the packet.}
	\label{fig:route}
	\vspace{-0.5cm}
\end{figure}

\subsection{Convergence-free Route Mechanism}
\label{subsec:design_route_management}

Anchors also intercept and reroute packets based on location information. As described in \S\ref{subsec:location_management}, users' IP addresses are allocated by satellite anchors while the fixed addresses of GSs are configured by the operator. Note that the addresses allocated by different anchors have different prefixes, which help \name to guide routing. Figure~\ref{fig:GS-to-MN} plots the route from a GS to user. When the ingress satellite of GS receives a packet to a user, it determines which satellite anchor manages the user according to the prefix and generates a tunnel to the anchor. When the anchor receives the packet, it finds the ingress satellite of the user (SAT in Figure~\ref{fig:GS-to-MN}) and generates a tunnel. Finally, SAT directly forwards the packet to the user by its forwarding information base (FIB) table. Further, Figure~\ref{fig:MN-to-GS} plots the route from a user to GS. Traditionally, the topology fluctuates due to handovers in GSs may trigger route recalculation, thus the user-to-GS traffic is interrupted. In \name, anchors also manage the GSs' locations to correctly forward packets without routing convergence. Specifically, when the ingress satellite receives a packet, the destination address of which is not allocated from any anchor, the satellite routes the packet to its anchor by default. Then the anchor determines the ingress satellite of the GS (SAT in Figure~\ref{fig:MN-to-GS}) and generates a tunnel. When the packet reaches SAT, it is forwarded to the GS by a direct route.
	\section{\name at the Network Level: Constellation-wide Anchor Management}
\label{sec:design_at_the_network_level}


\subsection{Understanding the Problem}
\label{subsec:design_formulation}

\noindent
\textbf{Problem formulation.} \revise{Let $S = \{s_1,s_2,...,s_m\}$, $U = \{u_1, u_2, ..., u_n\}$ and $G = \{g_1, g_2, ..., g_h\}$ denote a satellite set containing $m$ satellites, a user set containing $n$ users and a GS set containing $h$ GSs, respectively. Let the binary value $C_{i,j} = 1$ denote anchor $s_j$ manages $s_i$. Assume time is slotted as $\tau =\{t_1,t_2,...,t_T\}$. 
Let $\alpha_{i,j,k} = 1$ denote $s_j$ is visible to $u_i$ in $t_k$. $\alpha^{*}_{i,j,k} = 1$ means $u_i$ connects to $s_j$ in $t_k$. Similarly, $\beta_{i,j,k} = 1$ denotes $s_j$ is within $g_i$'s view in $t_k$ and $\beta^{*}_{i,j,k} = 1$ denotes $g_i$ connects to $s_j$ in $t_k$. Otherwise, $\alpha_{i,j,k}, \alpha^{*}_{i,j,k}, \beta_{i,j,k}$ and $\beta^{*}_{i,j,k}=0$. We use hops of the source and destination satellites to represent the latency between a communication pair. $D_{i,j}$ denotes the minimum hops between $s_i$ and $s_j$.}

Once an inter-cluster handover happens, the user needs to choose a new anchor and acquire a new address, which may interrupt upper-layer connections. The anchor manager aims to \emph{maximize the time users are managed by the same anchor} with various constraints by determining the anchor deployment and assignment. The anchor deployment and assignment~(ADA) problem can be formulated as follows.
\begin{alignat}{2}
\max \quad & \label{eq:objective}\revise{\sum_{i=1}^{n}\sum_{p,q,A=1}^{m}\sum_{k=1}^{T}\alpha^{*}_{i,p,k}\cdot\alpha^{*}_{i,q,k-1}\cdot C_{p,A}\cdot C_{q,A},}\\
\mbox{s.t.}\quad
&\label{eq:acc_constraint1}\revise{\alpha^{*}_{i,j,k} \le \alpha_{i,j,k}, \forall i \in [n]\footnotemark, j \in [m], k \in [T],}\\
&\label{eq:acc_constraint2}\revise{\beta^{*}_{i,j,k} \le \beta_{i,j,k}, \forall i \in [h], j \in [m], k \in [T],}\\
&\label{eq:delay_constraint}\alpha_{i,j,k}^*\cdot\beta_{i',j',k}^*\cdot[C_{j,A}\cdot(D_{j,A}+D_{A,j'})-D_{j,j'}]\le H \notag,\\
& \revise{\forall i \in [n], i' \in [h], j, j', A \in [m], k \in [T],}\\
& \label{eq:cluster_constraint}\revise{\sum_{j=1}^{m}C_{j,j'}=1, \forall j' \in [m].}
\end{alignat}

\footnotetext{$[n]$ is the set $\{i \in \mathbb{N} | i \le n\}$.}

\revise{The objective function~(\ref{eq:objective}) indicates the number of all tuples = $<i,p,q,k,A>$ satisfying $u_i$ connects to $s_p$ and $s_q$ in time $t_k$ and $t_{k-1}$, which belong to the same cluster managed by $s_A$.} The function minimizes the anchor switching times, \ie maximizes the CUR. Constraints~(\ref{eq:acc_constraint1}) and (\ref{eq:acc_constraint2}) indicate that a user or GS can only connect to visible satellites. Since the route between two satellites passes the anchor, we require that the maximum number of hops for a route in \name should not exceed the shortest path by more than $H$ hops to provide a low-latency services. Therefore, the constraint is formulated as constraint~(\ref{eq:delay_constraint}), \revise{which denotes the route from $s_{j}$ (ingress satellite of $u_i$) to $s_{j'}$ (ingress satellite of $g_{i'}$) via $s_A$ (anchor of $s_j$) is at most $H$ hops longer than the shortest path.} Finally, constraint~(\ref{eq:cluster_constraint}) denotes each satellite is managed by one anchor.

\begin{lemma}
\label{lemma:detour}
\revise{Given $s_i$ and anchor $s_A$, the maximum extra hops from any other satellite $s_j$ passing through $s_A$ to $s_i$ is $2 \cdot D_{A,i}$.}
\end{lemma}

\begin{proof}
Let $X$ and $Y$ denote the number of orbits and satellites in per orbit of the constellation, respectively. $s_i$ can be denoted as $(x_i,y_i)$, where $x_i$ is the orbit number and $y_i$ is the satellite number in an orbit. Obviously we have $D_{i,j}=\min(\left |x_i-x_j  \right |,X-\left |x_i-x_j  \right | )+\min(\left |y_i-y_j  \right |,Y-\left |y_i-y_j  \right | )$. \revise{Put three points $i,j$ and $A$ on a circle of circumference $X$ and let $\overset{\frown}{ab}$ denote the length of the shorter arc between point $a$ and $b$.} Thus $\overset{\frown}{ij}=\min(\left |x_i-x_j \right |,X-\left |x_i-x_j \right |)$. According to the relationship between the corresponding circle angles of each arc, we have $	\overset{\frown}{jA}-\overset{\frown}{ji} \le \overset{\frown}{Ai}$. We can get the same conclusion when the circumference is $Y$. Therefore, $D_{j,A}-D_{j,i} \le D_{A,i}.$ By augmenting both sides of the inequality with $D_{A,i}$, we have $D_{j,A}+D_{A,i}-D_{j,i}\le2 \cdot D_{A,i}$.
\end{proof}

\begin{theorem}
The ADA problem is NP-hard.
\end{theorem}

\begin{proof}
	
Assume users only see one satellite each time, thus $\alpha_{i,p,k}=\alpha_{i,p,k}^*$. We simplify this problem by setting $H$ to 2. According to lemma~\ref{lemma:detour}, the number of hops between any satellite and its satellite anchor should not exceed 1. In a +Grid topology, all candidate clusters satisfying the latency constraint can be enumerated in $O(m)$. For each candidate cluster $S_c$, $w(S_c)=\sum_{i=1}^{n}\sum_{p \in S_c}\sum_{k=1}^{T}\alpha_{i,p,k}-1$, where $w(S_c)$ is the weight of $S_c$. Then, By treating each candidate cluster as an item, its weight as an item value, and requiring that the items represented by two candidate clusters having common satellites cannot be taken together, ADA problem can be reduced to the disjunctively constrained knapsack~(DCK) problem~\cite{DCKP} with unlimited capacity in polynomial time, which is NP-hard.
\end{proof}

\setlength{\textfloatsep}{0.12cm}
\begin{algorithm}[tb]
	\caption{Pattern Discovery}\label{TD}  
	\textbf{Input:} visible satellite set within an expected time, $S^*$; \\maximum extra hop, \emph{$H$}\\ 
	\textbf{Output:} satellite cluster pattern, \emph{$f$}
	\begin{algorithmic}[1]
		\State{$s_A \gets -1,S' \gets \o,maxSize \gets 0 $}
		\For{$s_{i}$ in $S^*$}\Comment{\emph{Candidate satelite anchor}}
		\State{$S_{candidate}' \gets \o$}\Comment{\emph{Candidate satelite cluster}}
		\For{$s_{j}$ in $S^*$}
		\If{$2 \cdot D_{j,i} \le H$} \Comment{\emph{Lemma~\ref{lemma:detour}}}
		\State{$S_{candidate}'$.add$(s_j)$}
		\EndIf
		\EndFor
		\If{$|S_{candidate}'| \ge maxSize$}
		\State{$s_A \gets s_i, S' \gets S_{candidate}', maxSize \gets |S'|$}
		\EndIf
		\EndFor
		\State{$f \gets$ getPattern$(s_A, S')$}\\
		\Return{$f$}
	\end{algorithmic}
\end{algorithm}
\setlength{\floatsep}{0.12cm}

\subsection{Judiciously Deploying Anchors at LEO Satellites}
\label{subsec:design_algorithms}

We decompose the above problem to solve it efficiently. Specifically, we first (i) determine an anchor deployment and cluster division scheme, and (ii) design an algorithm for users to select anchors to maximize connection uninterrupted ratio.

Traditional anchor placement and service area division schemes in terrestrial mobile networks are often based on experience, such as the main geographical movement area of users (\eg a building). According to this, we get a key insight: \emph{traditional users' usual geographical movement area can be mapped to a set of satellites visible to the user within a certain period in ISTNs.} \revise{Therefore, we try to find a proper \emph{satellite cluster pattern}, which is defined as a mapping relationship from satellite nodes to satellite sets and denoted as $f(s_A) \to S_A$, where $s_A$ is the satellite anchor of satellite set $S_A$ and $S_A$ is also called an \emph{instance}. In particular, $f$ should satisfy that the relative positions between all satellites and the satellite anchor in each instance are fixed. To reduce the times of switching clusters, users should see at least one satellite in the instance for a long time. Then, we use as few instances as possible to cover all satellites in the entire constellation.}

\noindent
\textbf{Anchor deployment.} Algorithm~\ref{TD} introduces the details of finding a satellite cluster pattern from visible satellites in a certain period. For each candidate anchor, we calculate the maximum number of visible satellites satisfying the latency constraint (line 4-6). \revise{After traversing all candidate anchors, we adopt the largest candidate satellite cluster and get a pattern instance (line 7-8), from which we can derive the pattern (line 9).} Algorithm~\ref{APP} illustrates our heuristic anchor deployment algorithm. In each loop, we preferentially select a satellite as an anchor, whose cluster instance contains the most remaining satellites (line 3-8). At the end of each loop, satellites belonging to the instance are managed by this anchor (line 9-10).

\noindent
\textbf{Anchor assignment.} The key ideas of our greedy anchor assignment algorithm are users \revise{(i) keep connecting to connected satellites and satellites in the same cluster whenever possible }and (ii) select a new cluster that will be visible for the longest time in initial or when the previous cluster is out of the view.

\begin{theorem}
Given an anchor deployment and cluster division result, the greedy anchor assignment algorithm is optimal.
\end{theorem}

\begin{proof}
\revise{Let $s_{Alg(k)}$ denote the selected satellite anchor for a user in time $k$ with algorithm $Alg$. Denote the greedy algorithm is $Gr$ and the optimal algorithm closest to $Gr$ is $Opt$, which means that there is a maximum $K$ so that $Gr$ and $Opt$ choose the same satellite anchors in the first $K$ time slots. Then we have $s_{Gr(K)} = s_{Opt(K)}$, $s_{Gr(K+1)} \neq s_{Opt(K+1)}$ and $s_{Gr(K)} \neq  s_{Opt(K+1)}$. If $s_{Gr(K)} = s_{Gr(K+1)}$, then we can select $s_{Gr(K+1)}$ in $t_{K+1}$ based on $Opt$, so as to get an optimal solution $Opt'$ closer to $Gr$. If $s_{Gr(K)} \neq s_{Gr(K+1)}$ or $K = 0$, we can select $s_{Gr(K+1)}$ from $t_{K + 1}$ to $ t_{K + \delta }$ based on $Opt$ and get an optimal solution $Opt'$ closer to $Gr$, where $\delta$ is the longest time the cluster of $s_{Gr(K+1)}$ is visible to the user. Therefore, $Gr$ is proven to be optimal by the above contradiction.}
\end{proof}

\setlength{\textfloatsep}{0.12cm}
\begin{algorithm}[tb]
	\caption{Anchor Deployment and Cluster Division}\label{APP}  
	\textbf{Input:} satellite set, $S$; satellite cluster pattern, $f$\\ 
	\textbf{Output:} cluster division matrix, $C$
	\begin{algorithmic}[1]
		\State $R \gets S $ \Comment{\emph{$R$ records remaining satellites}}
		\State{$C \gets \boldsymbol{O}_{|S|,|S|}$}
		\While{$R\neq\varnothing$}
		\State{$maxCov \gets 0, Anc \gets -1$}
		\For{$s_A$ in $S$}
		\State{$Overlap \gets |f(s_A)$$\cap R|$}
		\If{$ Overlap \ge maxCov$}
		\State{$Anc \gets A,maxCov \gets Overlap$}
		\EndIf
		\EndFor
		\For{$s_i$ in $f(s_{Anc})$}
		\State{$C_{i,Anc} \gets 1$}	\Comment{\emph{Update cluster division}}
		\EndFor
		\State{$R \gets R-f(s_{Anc})$}
		\EndWhile\\
		\Return{$C$}
	\end{algorithmic}
\end{algorithm}
\setlength{\floatsep}{0.12cm}
	\section{Performance Evaluation}
\label{sec:evaluation}

We build a container-based ISTN simulator based on \cite{starrynet} as described in \S\ref{subsec:motivation_experiment_setup}. We conduct experiments based on data-driven simulations to explore the following aspects about \name: (i) can \name accomplish high connection uninterrupted ratio and low latency under both user movement and high LEO dynamics? (ii) how much network- and system-level cost does \name involve at satellite anchors?

\subsection{Experiment Setup}

\noindent

\noindent
\textbf{\name prototype.} Each satellite router builds one TCP connection with its corresponding anchor for reliable location updates. \revise{Anchors use DHCPv6\cite{dhcpv6} to allocate addresses for users.} \name leverages Segment Routing~(SR)~\cite{rfc8402} to realize the tunnel-based and convergence-free route mechanism mentioned in \S\ref{subsec:design_route_management}. We implement \name's deployment and assignment algorithms in Python. The anchor manager collects satellites' two-line element sets~(TLEs)\cite{TLE} from \cite{celestrak} to obtain constellation deployment and satellite trajectory data.

\noindent
\textbf{Constellation information.} Our simulation is based on the public information of two different constellations. The first is a complete shell of Starlink, including 1584 satellites in 72 orbits at an altitude of 540km~\cite{starlink_fcc}. The second is a shell in the Kuiper constellation which plans to deploy 1296 satellites in 36 orbits at 610km altitude~\cite{kuiper_fcc}. We also simulate scenarios with Inmarsat satellites which work at geostationary orbit and have provided commercial Internet services for airplanes~\cite{inmarsat}.

\noindent
\textbf{Model-based traffic generation.} For Starlink, we use geographic distribution of GSs published by~\cite{starlink_GS}. Since Amazon provides complete satellite services with their global infrastructure, we consider their 12 Amazon Web Services~(AWS) ground stations~\cite{Amazon_GS} and 105 available zones (AZs)~\cite{Amazon_server} as distributed GSs for Kuiper. For users, we follow the methodology proposed in \cite{pachler2021updated} which assumes 0.1\% of the population in each cell (1 latitude $\times$ 1 longitude) of the Starlink availability map\cite{map} are ISTN users. We also collect flight traces of two different airlines to mimic user mobility, one across the North Pacific and the other across the North Atlantic\cite{airline1,airline2}. Users in the evaluation access the Akamai content delivery network server~\cite{Akamai} closet to each GS.

\noindent
\textbf{Comparison.} We compare \ourname with Reliable Extreme Mobility Management (REM)\cite{REM}, ATOM\cite{ATOM} and Network-based Distributed Mobility Management (NDMM)\cite{DMM}. REM is designed for cellular networks based on ground anchors for user mobility. We extend REM in ISTNs by choosing a nearby GS as the anchor for mobile users. ATOM aims to select a better interface between cellular and WiFi networks for improve the quality of experience (QoE). In the evaluation of applying ATOM in ISTNs, anchors in ground stations will determine whether to forward traffic to the satellite or terrestrial networks to reduce latency. NDMM offloads the anchor functionalities to the network edge, which is considered as all satellites in the constellation. For \name, we limit $H$ (\ie the maximum extra hop as defined in \S\ref{subsec:design_formulation}) to $\frac{X+Y}{2}$, where $X$ and $Y$ is the number of orbits and satellites in per orbit, respectively.

\subsection{Connection Uninterrupted Ratio~(CUR)}
Figure~\ref{fig:constellation_reliability} plots the CUR in different constellations. The CUR only varies from 43.2\% to 79.3\% for both server-to-user and user-to-server traffic in REM and ATOM, because both data packets and location update messages suffer from the frequent and long-time routing instability due to handovers happening in GSs. NDMM can significantly improve the CUR in server-to-user traffic but only reaches about 60\% CUR in the user-to-server traffic. The main reason is that anchors on satellites in NDMM only manage the mobility of users and do not take into account the network location changes of GSs, thus making the user-to-server connection is also interrupted by route instability. \name reaches 98.7\% and 99.1\% of CUR in both directions of traffic. On the one hand, the location notification can be correctly and rapidly updated to the corresponding satellite anchor(s) due to stable ISL connections. On the other hand, anchors in \name provide convergence-free route mechanism by managing the location of network infrastructures (\ie GSs).

Figure~\ref{fig:airline_reliability} plots the CUR when users move in two different airlines. \revise{Firstly, the user-to-server CUR does not decrease significantly for all mechanisms although users are moving,} because this is mainly dependent on the space-ground routing instability, which does not change because the ground stations are still fixed. Secondly, the server-to-user CUR deceases by 4.3\% and 4.7\% on average for REM and ATOM. The main reason is that users' movement leads to an increase in the frequency of switching ingress satellites and ground stations (anchors), so the number of location update failures and the number of IP address changes also increase, which reduces the CUR. However, \name can still achieve more than 97.4\% CUR because users can still connect to satellites in the same cluster during their movements, which can keep the IP address unchanged. Finally, despite the use of geostationary satellites~(Inmarsat), which is stationary relative to the ground, the CUR is only 1.9\% to 2.4\% higher than \name, because the high speed of user movement still causes the change of anchor points in flight trajectories.

\begin{figure}[tbp]
	\centering
	\subfloat[CUR in Starlink (left) and Kuiper (right) topology.]{
		\begin{minipage}[t]{\linewidth}
			\centering
			\includegraphics[width=1.0\linewidth]{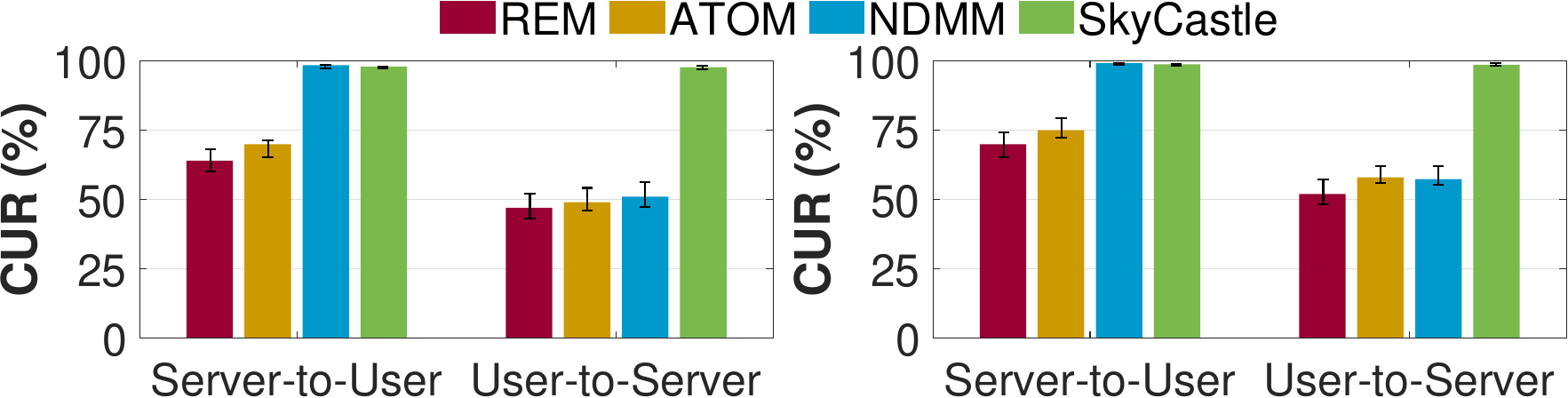}
			\label{fig:constellation_reliability}
			\vspace{-0.4cm}
		\end{minipage}%
	}

	\subfloat[CUR in Flight-United 857 (left) and Flight-Virgin Atlantic 138 (right).]{
		\begin{minipage}[t]{\linewidth}
			\centering
			\includegraphics[width=1.0\linewidth]{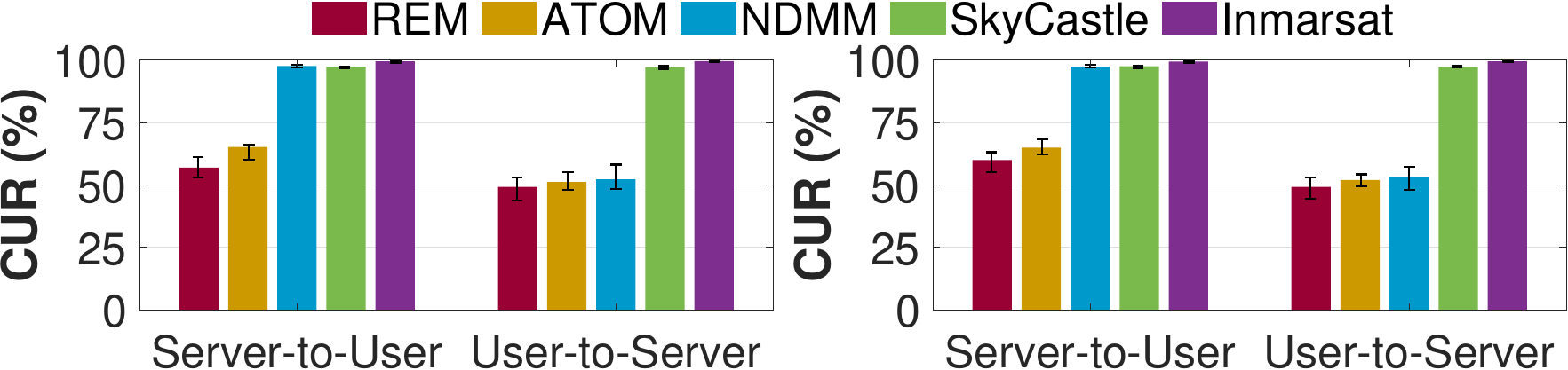}
			\vspace{-0.4cm}
			\label{fig:airline_reliability}
		\end{minipage}%
	}%
	\centering
	\caption{Connection uninterrupted ratio for different constellation topologies and flight trajectories.}
	\label{fig:network_reliability}
\end{figure}

\subsection{User-perceived Latency}

Figure~\ref{fig:constellation_RTT} plots the RTT between global users and servers in different constellations using different MM mechanisms. REM, ATOM, and NDMM all suffer from high latency due to the long distance between users and their anchors. Users in REM and ATOM naively registers with the nearest GS as the anchor, while the anchor in NDMM and \name is usually near the user's path to access the server. Therefore, unwise anchor selection can lead to serious detour and increase the latency. However, NDMM suffers from a high latency tail (up to 414.56ms in Starlink and 622.72ms in Kuiper) because the distances between anchor points on satellites and users increase rapidly over time. Although anchors are also deployed on satellites, \name reduces the latency by up to 47.8\% through changing anchors for users with a latency constraint.

\begin{figure}[tbp]
	\centering
	\subfloat[RTT in different constellation topologies.]{
		\begin{minipage}[t]{\linewidth}
			\centering
			\includegraphics[width=1.0\linewidth]{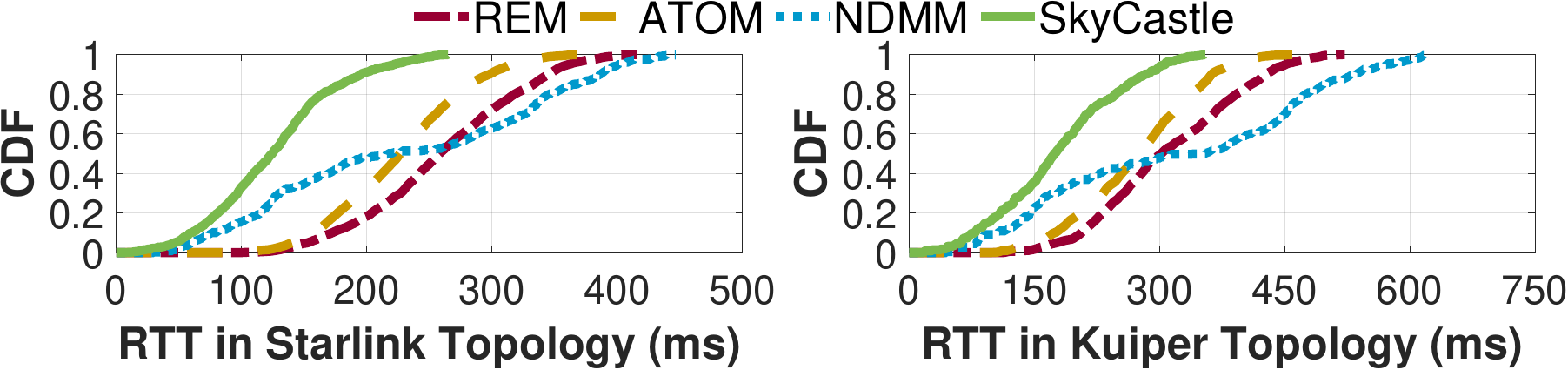}
			\vspace{-0.4cm}
			\label{fig:constellation_RTT}
		\end{minipage}%
	}
	
	\vspace{-0.2cm}
	\subfloat[RTT in different airlines.]{
		\begin{minipage}[t]{\linewidth}
			\centering
			\includegraphics[width=1.0\linewidth]{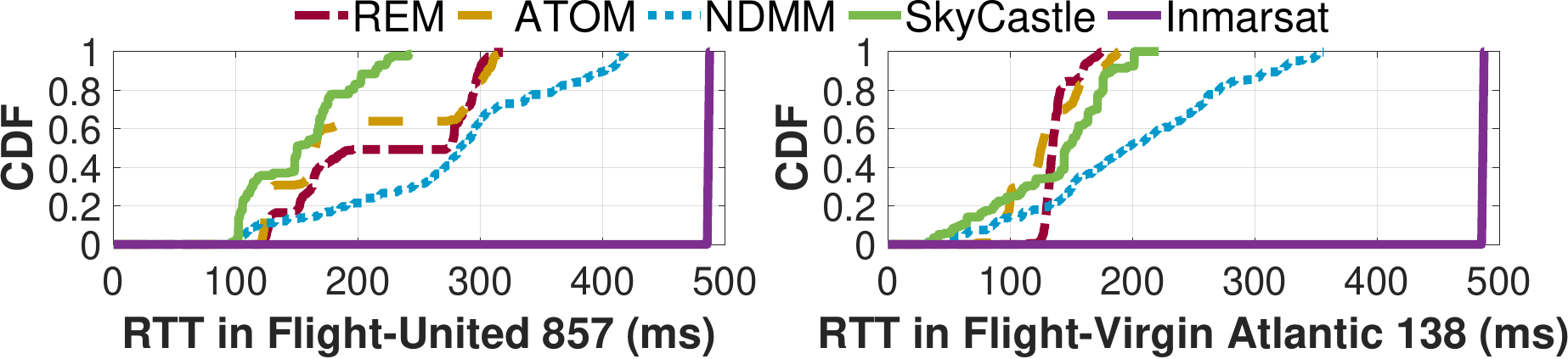}
			\vspace{-0.4cm}
			\label{fig:airline_RTT}
		\end{minipage}%
	}%
	\centering
	\caption{User-perceived RTT in global scale and different airlines.}
	\label{fig:delay}
\end{figure}

Figure~\ref{fig:airline_RTT} plots the RTT while users are moving in two different airplanes using different research mechanisms and an actual commercial technique. First, compared with Inmarsat, which consists of satellites working at the geostationary orbit, Starlink and Kuiper can reduce latency significantly by using LEO satellites. However, as the distance between the satellite anchor and the aircraft increases in NDMM, the highest latency has reaches 86.3\% of the latency of using Inmarsat. \name can achieve the lowest latency in the airline of United 857 than other mechanisms. The maximum latency in \name is slightly higher than that in REM and ATOM when users are moving at the airline of Virgin Atlantic 138. The primary reason for this is the proximity of GSs to the flight trajectory, which results in a shorter distance between the user's ingress satellite and the anchor in GS.

\subsection{Frequency of IP Address Update}

Table~\ref{tab:ip_change} shows the average number of IP address changes and handovers times per hour. Firstly, \name can reduce IP address change frequency by up to 69.1\% and 83.1\% in Starlink and Kuiper compared with REM and ATOM. The reason is that the high-speed movement of users is negligible for satellites, but travels for a long distance on the ground, thus connecting to many nearby GSs and reacquiring new IP address frequently. \revise{Besides, \name reduces handover times by about 23.5\% because users prefer connecting either to their current satellite or to new satellites within the same cluster.} Users do not change IP addresses in NDMM because they always consider the initial ingress satellite as anchors, which causes latency to increase by 46.6\% compared to \name.

\begin{table}[tb]
	\centering
	\small
	\resizebox{\columnwidth}{!}{
		\begin{tabular}{c|c|c|c|c|c}
			\hline
			\toprule[1pt]

			\textbf{Constellation}  & 
			\tabincell{c}{\textbf{Flight}}       & \tabincell{c}{\textbf{REM} \\ {\cite{REM}}}    &  \tabincell{c}{\textbf{ATOM} \\ \text{\cite{ATOM}}}    & \tabincell{c}{\textbf{NDMM} \\ \text{\cite{DMM}}} & \tabincell{c}{\name \\ (this paper)}     \\
			\midrule[0.5pt]
				\multirow{2}{*}{\tabincell{c}{\textbf{Starlink} \\ \textbf{(201 GSs)}}} &  \textbf{United 857}      &  5.5/50.3   &   5.5/47.3        &   0/50.3  &  1.7/39.8 \\
			& \textbf{Virgin Atlantic 138}   & 3.2/54.7     &  3.2/51.2        &     0/54.7  &  1.6/41.5   \\
			\midrule[0.5pt]
			\multirow{2}{*}{\tabincell{c}{\textbf{Kuiper} \\ \textbf{(117 GSs)}}} & \textbf{United 857}      &  6.5/38.5   &   6.5/34.6        &   0/38.5  &  1.1/33.9 \\
			& \textbf{Virgin Atlantic 138}   & 2.4/43.5     &  2.4/37.2       &     0/43.5  &  1.1/34.8   \\
			\bottomrule[1pt]
			\hline
	\end{tabular}}
	\caption{Average number of IP changes / handovers per hour.}
	\label{tab:ip_change}
\end{table}

\subsection{Cost Analysis}

For evaluating more realistic system-level overhead, we also deploy various mechanisms on a Raspberry Pi 4 computer and replace a container in the simulator, which has been tested in a real space environment, to demonstrate the advantages of \name. Figure~\ref{fig:overhead} plots the network-level and system-level overhead while running different mechanisms.

Figure~\ref{fig:ctrl} plots the control overhead required by different mobility management mechanisms, which is defined as the number of transmission hops taken by all control messages per second. We divide the control overhead into location and route management overhead. Firstly, the location management overhead of \name is lower than other mechanisms because of the reduction of handover times and the transmission of most update messages are limited within a cluster. Secondly, the high route management overhead of REM, ATOM, and NDMM comes from the route update messages caused by the handovers between GSs and satellites. However, in \ourname the route recovery after handovers in GSs can also be solved by the location management and route mechanism of satellite anchors without route convergence. Therefore, there is no additional route overhead in the same simulated scenario.

Figure~\ref{fig:CPU} plots the CPU usage of the Raspberry Pi 4 while running different mechanisms. In REM, ATOM, and NDMM, satellites suffer from higher CPU usage because they are required to handle route calculation and update users' location simultaneously. Specially, in NDMM, the satellite may also act as an anchor to perform MM functions. In \ourname, as similar to NDMM, the satellite anchor needs to handle location update messages and maintain a location database. However, there is no need for frequent route calculation triggered by handovers, which significantly reduces the CPU usage.

\begin{figure}[tbp]
	\centering
	\subfloat[Control overhead.]{
		\begin{minipage}[t]{0.48\linewidth}
			\centering
			\includegraphics[width=1.0\linewidth]{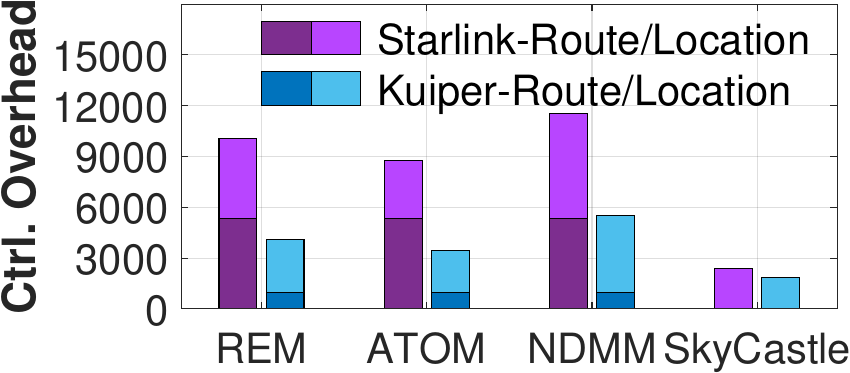}
			\label{fig:ctrl}
			\vspace{-0.2cm}
		\end{minipage}%
	}\hspace{0.15cm}
	\subfloat[CPU usage.]{
		\begin{minipage}[t]{0.46\linewidth}
			\centering
			\includegraphics[width=1.0\linewidth]{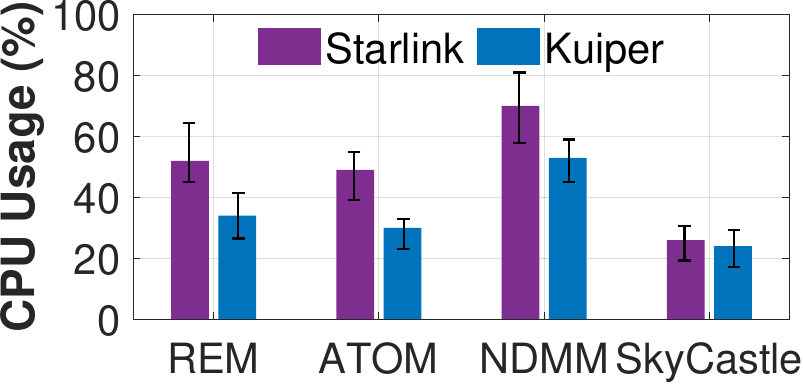}
			\label{fig:CPU}
			\vspace{-0.2cm}
		\end{minipage}%
	}%
	\centering
	\caption{Network- and system-level overhead.}
	\label{fig:overhead}
\end{figure}
	\section{Related Work}
\label{sec:related_work}

\noindent
\textbf{Anchor-based IP mobility management.} The IETF work groups have standardized many RFCs for using anchors to handle mobility management problem~\cite{MIP,HMIPv6,FMIPv6,PMIPv6,PFMIPv6,DMM} in WiFi networks. In cellular networks, anchors also play a big role in the evolution of communication technologies (\eg from 4G to 5G)~\cite{li2016instability,hassan2022vivisecting}. UbiFlow~\cite{ubiflow} uses distributed controllers to manage mobility of devices in the software defined Internet of things~(IoT). Recent works\cite{SIGMA,DIP} attempt to extend these solutions to ISTNs by setting location management functions at GSs. Due to a lack of consideration for the network infrastructure mobility, location updates may fail, leading to low availability. In \cite{HiMIPv6}, a user sets the first access satellite as the anchor. However, authors do not consider the mobility of anchors, reducing other network performance.

\noindent
\textbf{Anchorless IP mobility management.} NDM proposed in \cite{NDM} reduces the overhead of mobility management by managing users' handover in batches, but users need to reacquire IP addresses after handovers. There has been an attempt to break the constraints of binding identifier and locator in one IP address to support mobility management~\cite{HIP,ILNP} but fails to be popularized due to high deployment cost. This idea inspired recent efforts\cite{LISP_LEO}, which embed user's geo-location in IP address to reduce the impact of mobility. However, a geo-location may be covered by multiple satellites, and this last-hop ambiguity problem leads to routing failure.
	\section{Conclusion}
\label{sec:conclusion}

Efficiently managing global mobility is a crucial issue for satellite operators. Compared to traditional mobility management, network-layer mobility management in ISTNs presents unique challenges. In this paper, we propose a novel approach by deploying anchors in fast-moving satellites. We introduce \name, a global mobility management mechanism for ISTNs that facilitate seamless and low-latency Internet services with decent network performance and acceptable control overhead. Trace-driven evaluations demonstrate that our solution can significantly improve connection uninterrupted time by up to 55.8\% and reduce latency by 47.8\%.

\section{Acknowledgment}
This work was supported by the National Key R\&D Program of China (No. 2022YFB3105202) and National Natural Science Foundation of China~(NSFC No. 62372259 and No. 62132004).
	
	\bibliographystyle{unsrt}
	\bibliography{skycastle}
	
\end{document}